\documentclass[12pt,a4paper]{article}
\usepackage[centertags]{amsmath}
\usepackage{amsfonts,amsthm,amssymb}
\usepackage{amssymb}
\usepackage{amsmath}
\usepackage{graphicx}

\def\ne{\boldsymbol{WSN}}
\def\ane{\boldsymbol{ANE}}
\def\ned{\boldsymbol{WSN}_{dist}}
\def\fp{\boldsymbol{AFP}}
\def\fpd{\boldsymbol{AFP}_{dist}}
\def\ep{\boldsymbol{ESP}}
\def\hp{\boldsymbol{HTP}}
\def\p{\boldsymbol{\Pi}}
\def\E{\mathbb{E}}
\theoremstyle{definition}
\newtheorem{theorem}{Theorem}
\newtheorem{corollary}{Corollary}
\newtheorem{definition}{Definition}
\newtheorem{lemma}{Lemma}
\newtheorem{proposition}{Proposition}
\newtheorem{remark}{Remark}

\newtheorem{example}{Example}

\begin{document}

\title{Query Complexity of Approximate Nash Equilibria}
\author{Yakov Babichenko\footnote{Center for the Mathematics of Information, Department of Computing and Mathematical Sciences, California Institute of Technology. E-mail: babich@caltech.edu.}\footnote{The author wishes to thank Noam Nisan, Sergiu Hart, Paul Goldberg, Yishay Mansour, and Stephen Vavasis for useful discussions and comments. The author  gratefully acknowledges support from a Walter S. Baer and Jeri Weiss fellowship.}}

\maketitle

\begin{abstract}
We study the query complexity of approximate notions of Nash equilibrium in games with a large number of players $n$. Our main result states that for $n$-player binary-action games and for constant $\varepsilon$, the query complexity of an $\varepsilon$-well-supported Nash equilibrium is exponential in $n$. One of the consequences of this result is an exponential lower bound on the rate of convergence of adaptive dynamics to approxiamte Nash equilibrium.
\end{abstract}

\section{Introduction}

The problem of computing Nash equilibrium is known to be hard (see \cite{DGP}), even for two-player games (see \cite{CD}). However, there are still many open questions regarding the complexity of an \emph{approximate} Nash equilibrium. In this paper we will focus on one of them. Given a normal form $n$-player game with a constant number of actions $m$ for each player, how hard is it to compute an approximate Nash equilibrium? Note that in the above problem the size of the input is exponential. A reasonable model to consider in such a case is the \emph{query model}. Instead of a huge input that contains the whole game, we assume the existence of a \emph{black box}. For every query of the algorithm about payoffs in the game, the black box returns an answer. The queries could be about either action profiles (see Section \ref{sec:main}) or distributions over action profiles (see Section \ref{sec:dist}). We measure the complexity of an algorithm by the number of queries that it asks for the worst case input.

Before stating our main result on the query complexity of approximate \emph{Nash} equilibrium, we introduce the state of the art for the related question on the query complexity of approximate \emph{correlated} equilibrium. There exists a \emph{randomized} algorithm for computing an approximate correlated equilibrium using only $poly(n)$ payoff queries. Such a surprising\footnote{The result is surprising because after $poly(n)$ queries the algorithm knows only a tiny fraction of payoffs in the game ($\frac{poly(n)}{m^n}$). Nevertheless, the algorithm knows (with high probability) that a certain distribution forms an approximate correlated equilibrium.} result is achieved by \emph{regret minimizing algorithms} (for instance the \emph{regret matching} algorithm; see \cite{HN},\cite{HMc1},\cite{HMc4}). On the other hand, if we restrict the algorithm to be \emph{deterministic}, then the computation of an approximate correlated equilibrium requires exponential (in $n$) number of queries. This result is proved in a recent paper by Hart and Nisan \cite{HN}. Obviously, this result induces an impossibility result also for the computation of approximate Nash equilibrium. Namely, for deterministic algorithms the computation of an approximate Nash equilibrium requires exponential (in $n$) number of payoff queries. But it leaves open the question on the query complexity of approximate Nash equilibrium for \emph{randomized} algorithms. This question was posed in \cite{HN} as an open problem. In this paper we answer this question: the query complexity of an approximate Nash equilibrium is exponential, even for randomized algorithms (see Theorems \ref{theo:main}, \ref{theo:bin}, \ref{theo:NEdist}, and \ref{theo:ane}). 

We discuss two notions of additive approximation of equilibrium: approximate Nash equilibrium and approximate well-supported Nash equilibrium (see Section \ref{sec:games} for definition and a short discussion on these two notions). Our main result (Theorem \ref{theo:bin}) states that even for a constant approximation size, the query complexity of an approximate well supported Nash equilibrium in $n$-player binary-action games is $2^{\Omega(n)}$. The proof of the result is based on a novel reduction from the problem of finding a fixed point of a function to the problem of finding an approximate Nash equilibrium. The connection between the hardness of computing Nash equilibria and the hardness of computing fixed points was previously established (see \cite{DGP} and \cite{CD}). However, the established reductions fail to work for an approximate Nash equilibrium for \emph{constant} approximations. Beyond the simplicity of our reduction (see Section \ref{sec:ne->fp-const}), to the best of my knowledge this is the first reduction that works for an approximate notion of Nash equilibrium (the approximate well-supported Nash equilibrium) with \emph{constant} approximation.

The exponential lower bound on the query complexity of \emph{approximate} Nash equilibria has several applications and consequences.




1. For computational complexity, this result provides evidence that it is very unlikely that there is a polynomial (in $n$) algorithm for computing an approximate Nash equilibrium in $n$-player binary-action games. This is because if such an algorithm exists, then it must depend on more complex data of the game than just expected payoffs under distributions. Note that for the parallel question of computing an approximate or even an exact correlated equilibrium, such an algorithm exists (see \cite{JL}, and \cite{PR}).

2. The result provides insights into the rate of convergence of adaptive dynamics to an approximate Nash equilibrium (see Section \ref{sec:dyn}). The question of the convergence of adaptive dynamics to an exact Nash equilibrium (pure or mixed) was studied by Hart and Mansour \cite{HM}, where they provide exponential lower bounds via communication complexity results. However, the problem of convergence to an \emph{approximate} Nash equilibrium remained open. Our result yields an exponential lower bound on the rate of convergence of adaptive dynamics to an \emph{approximate} Nash equilibrium (see Corollary \ref{cor:dyn-pure} and Theorem \ref{theo:dyn}) for a wide class of adaptive dynamics, which we call $k$-queries dynamics (see Definition \ref{def:qd}).

3. The third consequence follows from the proof of the main Theorem. After reducing the problem of finding approximate well-supported Nash equilibrium to the problem of finding approximate fixed point of a function, we analyze the query complexity of finding approximate fixed point of a function for \emph{randomized} algorithms. The complexity of computing approximate fixed-point is of an independent interest and was previously studied in the context of query complexity. Hirsch, Papadimitriou, and Vavasis \cite{HPV} studied it for \emph{deterministic} algorithms, whereas we analyze it for \emph{randomized} algorithms. These more general settings allow us to answer an open question that was posed in \cite{HPV} twenty five years ago: what is the query complexity of an approximate fixed point in the case where queries are \emph{distributions over the domain} (rather than just points in the domain). Theorem \ref{theo:fp}, answers this question: even if the queries are distributions, an $\exp(n)$ number of queries is needed to find an approximate fixed point.

In addition to the above mentioned literature, there are few other recent papers that study the query complexity of equilibria.

Fearnley, Gairing, Goldberg, and Savani \cite{FGGS} study the query complexity of an approximate Nash equilibrium. In particular, they derive a lower bound on the complexity of an approximate Nash equilibrium in two-player games. In addition, they provide several classes of $n$-player games where a polynomial number of queries is enough to find an approximate Nash equilibrium. The negative result of this paper shows that unlike the above mentioned games, in general games an exponential number of queries is needed.

Goldberg and Roth \cite{GR} study the query complexity of several notion of approximate equilibria. In particular, they prove that for games with succinct representation size $p$, the query complexity of approximate well-supported Nash equilibrium is polynomial in $n$ (number of players), $m$ (number of actions for each player) and $p$ (the representation size). This result demonstrates that the exponential lower bound presented in the present paper cannot hold for succinctly representable games. We point out that the query algorithm of Goldberg and Roth \cite{GR} for computing an approximate well-supported Nash equilibrium is \emph{not} computationally efficient. Moreover, in a recent paper, Rubinstain \cite{R1} proves that computing an approximate well-supported Nash equilibrium is PPAD-hard for succinctly representable games with $n$ players and constant number of actions. His proof is based on our novel reduction (Section \ref{sec:ne->fp-const}) from the fixed point problem to the approximate Nash equilibrium problem, which holds for constant $\varepsilon$. In an even more recent paper, Rubinstain \cite{R2} proves that this PPAD-hardness result actually holds for a very simple class of succinctly representable games: graphical polimatrix games.

\section{Preliminaries}

\subsection{The query model}

In the query model every problem $\p$ is specified by possible inputs, desired outputs, and \emph{queries}. The queries are specified by the type of questions that can be asked, and by the answers that are provided. A \emph{query algorithm}, which in this paper will be called simply \emph{algorithm}, is a procedure that asks queries in an adaptive manner and essentially for every input generates an output. Note that we put no computational constraints on the way the algorithm generates the next query or the output, given the previous answers.


For randomized algorithms, we allow errors in the output. Namely, we require that for all inputs the answer will be correct only with probability $p<1$. 


Given an input, the number of queries of a randomized algorithm is a random variable. There are two reasonable definitions for the \emph{query complexity of a probabilistic algorithm}. One definition is via expectation, namely, the expected number of queries for outputting the correct answer with probability $p$, which is denoted by $QC_{\E,p}(\p)$. Another definition restricts the number of queries to be at most $T$, and we ask what is the minimal number $T$ such that there exists an algorithm that outputs the correct answer with probability $p$. We denote this number of queries by $QC_p(\p)$. There is a close relation between these two definitions. One such relation is given in the following remark. 

\begin{remark}\label{rem:qq}
Note that $QC_p(\p) \leq 2 QC_{\E,2p}(\p)$, because for an algorithm with expected query complexity $q$ that outputs the correct answer with probability $2p$, we can stop the algorithm after $2q$ steps. Then, by Markov inequality, the output will be correct with probability of at least $\frac{1}{2}(2p)$.
\end{remark}

The results in the present paper are lower bounds. We will formulate the results using $QC_p$. By the above remark we can easily translate all these lower bounds on $QC_p$ into lower bounds on $QC_{\E,2p}$. 

\subsection{Normal form games}\label{sec:games}
We consider normal form $n$-player games where every player has $m$ actions, and the payoffs are in $[0,1]$. We use the standard notations. The \emph{set of players} is $[n]$. The \emph{set of actions of player} $i$ is $A_i$. The \emph{set of action profiles} is $A=\times_i A_i$. The \emph{payoff function} of player $i$ is $u_i:A \rightarrow [0,1]$. The set of \emph{mixed strategies} of player $i$ is denoted by $\Delta(A_i)$. For a mixed strategy $x_i \in \Delta(A_i)$ we denote by $supp(x_i)\subset A_i$ the set of strategies that are played with positive probability. The set of \emph{distributions over action profiles} is denoted by $\Delta(A)$. The payoff function $u_i$ can be multilinearly extended into $u_i:\Delta(A) \rightarrow [0,1]$. The \emph{payoffs profile} is denoted by $u=(u_i)_{i=1}^n$, and we will identify the game with $u$.

Given a profile of mixed actions $x=(x_i)_{i=1}^n$ where $x_i \in \Delta(A_i)$, we denote by $br_i(x)=max_{a_i \in A_i} u_i(a_i, x_{-i})$ the best-reply value of player $i$, i.e., the maximal payoff that player $i$ can get against opponents' strategy $x_{-i}$.
 
There are two different notions for an additive approximation for Nash equilibrium.

The first one, an $\varepsilon$\emph{-Nash equilibrium} ($\varepsilon$-NE for short), requires that every player receives a payoff of at least $br_i(x)-\varepsilon$, i.e., $u_i(x)\geq br_i(x)-\varepsilon$ for every player $i$. Namely, \emph{the mixed strategy of a player} leads to a high enough payoff.

The second one, which is not less intuitive, is an $\varepsilon$-\emph{well-supported Nash equilibrium} ($\varepsilon$-WSNE for short) which requires that every player assign positive probability only to actions which lead to a payoff of at least $br_i(x)-\varepsilon$, i.e., if $a_i \in supp(x_i)$ then $u_i(a_i,x_{-i})\geq br_i(x)-\varepsilon$. Namely, \emph{each action that a player plays}, leads to a high enough payoff.

Note that every $\varepsilon$-WSNE is an $\varepsilon$-NE, but not vice versa. Nevertheless, we can construct an approximate WSNE from an approximate NE by relaxing the approximation from $\varepsilon$ to $\Theta(\sqrt{\varepsilon}n)$. Such a construction appears in Daskalakis, Goldberg, and Papadimitriou \cite{DGP}, Lemma 4.28, and a variation of this construction appears in the proof of Theorem \ref{theo:ane} in the Appendix A. 

Our main focus in this paper will be on approximate well supported Nash equilibria. But it will induce results also about approximate Nash equilibria (see Theorem \ref{theo:ane}).

\section{The Results}\label{sec:res}

\subsection{The main results}\label{sec:main}

Consider the problem of an approximate well-supported Nash equilibrium.

$\ne(n,m,\varepsilon)$:

INPUT- $n$-player game $u$ where every player has $m$ actions and the payoffs are in $[0,1]$.

OUTPUT- An $\varepsilon$-well-supported Nash equilibrium.

QUERIES- Each query is a pure action profiles $a$ and the answer is the payoffs profile $u(a)$.

We show in Theorem \ref{theo:bin} that even for the case of $m=2$ and constant $\varepsilon$ the problem $\ne$ requires $2^{\Omega(n)}$ queries. Before that, we introduce in Theorem \ref{theo:main} slightly weaker result: for constant $m$ and constant $\varepsilon$ the problem $\ne$ requires $2^{\Omega(n)}$ queries. The reasons for including the weaker result are the following two. First, the proof of the stronger result (for the case of $m=2$) is based on the proof of the weaker result (for the case of constant $m$) which is cleaner and simpler for understanding. Second, the lower bound for the case of constant $m$ is $2^{n/6}$ (if we ignore polynomial factors) whereas the lower bound for the case of $m=2$ is $2^{n/22170}$. Namely, the lower bound for the case of constant $m$ is significantly better (although both of them are exponential).
    
\subsubsection{Constant number of actions for each player}


\begin{theorem}\label{theo:main}
Fix $m=3609$ and $\varepsilon=\frac{1}{2}10^{-7}$. For every probabilistic algorithm that uses $16 \cdot 2^{n/6}/n^4$ pure-action queries to compute an $\varepsilon$-well-supported Nash equilibrium in $n$-player games with $m$ actions for all players, there exists a game where it returns a correct answer with probability of at most $3\cdot 2^{-n/6}$. I.e., 
\begin{equation*}
QC_p(\ne(n,m,\varepsilon))\geq 16 \frac{2^{\frac{n}{6}}}{n^4}=2^{\Omega(n)}
\end{equation*}
for $p=3\cdot 2^{-n/6}$.
\end{theorem}

This theorem is in contrast to the correlated equilibrium case, where the regret-minimizing algorithms (see, e.g., \cite{LW} and \cite{HMc1}) require only a polynomial number of queries to find an approximate correlated equilibrium.

%

The complete proof of Theorem \ref{theo:main} appears in Section \ref{sec:proof}. We present here a brief outline of the proof.

\textbf{Outline of the proof of Theorem \ref{theo:main}.} The proof is done in three steps. 

In the first step (Section \ref{sec:ne->fp-const}), we reduce the problem of finding approximated WSNE in a $2n$-player game to the problem of finding an approximate fixed point of a Lipschitz continuous $n$-dimensional mapping. For every function $f$ we define a game $u=u(f)$ (see equations (\ref{eq:payoff-def1}) and (\ref{eq:payoff-def2})), such that every approximate WSNE of $u(f)$ corresponds to an approximate fixed point of $f$. This reduction is based on a proof of Brower's fixed-point theorem using Nash's theorem. The proof appears in a blog by Eran Shmaya \cite{S}.

In the second step (Section \ref{sec:fp->ep}), we introduce the reduction of Hirsch, Papadimitriou, and Vavasis \cite{HPV} from the problem of finding an approximate fixed point of an $n$-dimensional mapping to the problem of finding the end of a simple path (path with no cycles) on the $n$-dimensional hypercube. For every simple path $\alpha$, they construct in a mapping $f=f(\alpha)$ such that every fixed point of $f$ corresponds to the end of the path $\alpha$ (see properties (P1)--(P3) in Remark \ref{rem:pro}).

Finally, in the third step (Section \ref{sec:ep->hp}), we prove that the end of a simple path is a hard problem, even in probabilistic settings. Hart and Nisan \cite{HN} show that the end of a path (not necessarily a simple one) is a hard problem. Using similar arguments to those in \cite{HN}, we strengthen this hardness result: it is hard to find the end of a path, even if it is known that the path is simple.

%
%
%

\subsubsection{Binary-action games}

Theorem \ref{theo:main} proves an exponential lower bound on the number of queries that are required for finding an approximate WSNE for games with constant ($m=3609$), but huge, number of actions for each player. Our second main theorem states that an exponential lower bound holds even for the case where each player has only two actions ($m=2$).

\begin{theorem}\label{theo:bin}
There exist constants\footnote{We prove the result for the constants $c_1=c_2=\frac{1}{22170}$ and $\varepsilon=\frac{1}{14\cdot 10^6}$.} $c_1,c_2,\varepsilon>0$ such that for $p=2^{-c_1n}=2^{-\Omega(n)}$ holds
\begin{align*}
QC_p(\ne(n,2,\varepsilon))\geq 2^{c_2n}=2^{\Omega(n)}.
\end{align*}
\end{theorem}

The proof is relegated to Section \ref{sec:proof}. The idea is to modify  the reduction from the approximate fixed point problem to approximate WSNE that appears in the proof of Theorem \ref{theo:main}. We introduce a version of this reduction that holds for \emph{binary-action} games. Roughly speaking, the idea is to replace a player with $m$ actions by $m$ different agents where each agent has only two actions. The specific structure of the payoffs in the reduction of Theorem \ref{theo:main} allows us to do so, without changing the property that every approximate WSNE of the constructed game corresponds to an approximate fixed point of the function.

\begin{remark}
In the remainder of the paper, we will state all the consequences from the main theorems for binary action games rather than for games with constant number of actions, simply because this result is theoretically stronger.

All the consequences can be stated also for the case of constant number of actions (more precisely, for games with $3609$ actions for each player), and then the constant at the exponent of the lower bound is significantly better ($\frac{1}{6}$ instead of $\frac{1}{22170}$).
\end{remark}

\begin{remark}
For ease of presentation, in the remainder of the paper, we will not explicitly mention the exact constants in the theorem statements. Our results are essentially asymptotic in nature. The exact constants that follow from our proofs are poor (for example, we prove a lower bound of $2^{cn}$ for a very small value of $c$). Improving the underlying constants remains an open question. 
\end{remark}

\subsection{Distribution queries}\label{sec:dist}

In Theorem \ref{theo:bin} we considered the model where each query is a pure action profile $a\in A$ in the game. We would like to generalize the exponential lower bound of Theorem \ref{theo:bin} to the case where each query can be a \emph{distribution} over action profiles $x\in \Delta(A)$. 

The most natural model that comes to mind is the one where the answer to the query $x$ is the exact value $u(x)$. This model is not so interesting. In this model \emph{one} query is enough to receive the complete information about all the payoffs in the game. We illustrate this fact by an example.

\begin{example}\label{ex}

Assume that all payoffs are $0$ or $1$. We numerate all the action profiles in the game by $A=\{a(0),a(1),...,a(N-1)\}$ where $N=m^n$, and we query the distribution $x$ with 
\begin{equation*}
x(a(j))=\frac{2^j}{2^N-1}.
\end{equation*}
From the answer
\begin{equation*}
u(x)=\frac{1}{2^N-1}\sum_{j=0}^{N-1} u(a(j))2^j
\end{equation*}
the algorithm can derive all the values $u(a(j))$.

\end{example}

This example can be easily generalized from the $0,1$-payoffs case to the case where all payoffs are from the set $\{k/M:0\leq k \leq M, k \in \mathbb{N} \}$ for every $M$.

The exact-answer model mentioned above is not so interesting because the answer $u(x)$ may contain a huge amount of information, as illustrated in Example \ref{ex}.

A reasonable way to overcome this issue is to assume that the answers are given \emph{with precision $\delta$.} Namely, for every query $x\in \Delta(A)$ the answer is \emph{some} vector $w\in \mathbb{R}^n$ where $||w-u(x)||_\infty <\delta$. In this model we consider the problem of approximate WSNE. 

$\ned(n,m,\varepsilon,\delta)$:

INPUT- $n$-player game $u$ where every player has $m$ actions and the payoffs are in $[0,1]$.

OUTPUT- An $\varepsilon$-well-supported Nash equilibrium.

QUERIES- Each query is a distribution over actions $x\in \Delta(A)$ and the answer is some payoff vector $w$ where $||w-u(x)||_\infty \leq \delta$.

We will say that an algorithm \emph{solves the problem} $\ned$ in $T$ queries $(x_t)_{t=1}^T$ if it outputs the correct answer for every sequence of answers $(w_t)_{t=1}^T$ where $||w_t-u(x_t)||_\infty < \delta$ for all $t\in [T]$.

For every distribution query $x$, the answer $u(x)$ can be well approximated by a long enough sequence of pure action queries, simply by sampling from the distribution $x$ in an i.i.d. manner (see \cite{LMM}). Hence the impossibility result of theorems \ref{theo:main} and \ref{theo:bin} yields an impossibility result also for the distribution model.

\begin{theorem}\label{theo:NEdist}
$QC_p(\ned(n,2,\varepsilon,\delta))= \delta^2 2^{\Omega(n)}$ for constant $\varepsilon$ and for $p=2^{-\Omega(n)}$.



\end{theorem}

Theorem \ref{theo:NEdist} yields that even for answers that are given with exponentially small precision (i.e., $\delta=2^{-\Omega(n)}$), the exponential lower bound still holds.

Theorem \ref{theo:NEdist} emphasizes even more the difference from the correlated equilibrium case. Following Jiang and Layton-Brown \cite{JL}, an \emph{exact} correlated equilibrium can be computed using a polynomial number of distribution queries (see also \cite{BB}). By Theorem \ref{theo:NEdist}, then, not only an exact Nash equilibrium, but even an approximate Nash equilibrium cannot be computed.

\begin{proof}[Proof of Theorem \ref{theo:NEdist}]
Every $\ned$ algorithm that uses $\delta^2 2^{cn}/n $ distribution queries with success probability $p+ \delta^2 2^{(c-2)n}$ induces a $\ne$ algorithm with $2^{cn}$ queries with success probability $p$. We replace every distribution query $x$ by $n/\delta^2$ pure action queries that are sampled i.i.d. from $x$. By Hoeffding's inequality (see \cite{H}), the probability that the sample approximates $u_i(x)$ with a precision of $\delta$ is at least $1-2e^{-2n}>1-2^{-2n}$. Therefore the probability that all the $\delta^2 2^{cn}/n$ queries will be well approximated for all players is at least $1- \delta^2 2^{cn} 2^{-2n}=1-\delta^2 2^{(c-2)n}$.

By Theorem \ref{theo:bin} for constant $c<2$ and $p=2^{-\Omega(n)}$ there is no $\ne$ algorithm with $2^{cn}$ queries with success probability $p$. This implies that for every algorithm that uses $\delta^2 2^{cn}/n =\delta^2 2^{\Omega(n)}$ distribution queries the probability of success is at most $p+ \delta^2 2^{(c-2)n}=2^{-\Omega(n)}$.
\end{proof}

\subsection{Approximate Nash equilibrium}\label{sec:fp}

The approximate (not necessarily well-supported) Nash equilibrium problem is denoted by $\ane$. $\ane$ has the same input and the same queries as the $\ne$ problem. The desired output is an $\varepsilon$-Nash equilibrium. For the approximate Nash equilibrium case, Theorem \ref{theo:main} induces the following result.

\begin{theorem}\label{theo:ane}
$QC_p(\ane(n,2,\frac{1}{n}))\geq 2^{\Omega(n)}$ for $p=2^{-\Omega(n)}$.
\end{theorem}

This theorem excludes the possibility of a sub-exponential full approximation scheme for the Nash equilibrium in the query model.

The idea of the proof of Theorem \ref{theo:ane} is simple, and is presented below. The formal proof is slightly technical and it is relegated to the Appendix.

\textbf{Outline of the proof of Theorem \ref{theo:ane}.} The idea is that we can construct $O(\varepsilon)$-WSNE $(y_i)_{i=1}^n$ from an $(\varepsilon^2/n)$-Nash equilibrium $(x_i)_{i=1}^n$. Once we prove this, the lower bound $QC_p(\ane(n,2,\frac{\varepsilon^2}{n}))\geq 2^{\Omega(n)}$ follows immediately from Theorem \ref{theo:bin}.

Daskalakis, Goldberg, and Papadimitriou \cite{DGP}, Lemma 4.28, present such a construction. For every player $i$ let $a_i^*$ be one of the best replies to $x_{-i}$, and let $br_i=u_i(a_i^*,x_{-i})$ be the best-reply value. Fix the threshold $t_i=br_i-\varepsilon$. In the mixed strategy $y_i$, every probability mass on an action $a_i$ such that $u_i(a_i,x_{-i})<t_i$ is replaced by a probability mass on the action $a_i^*$.

We cannot use this construction directly, because in the pure-action queries model the values $u_i(a_i,x_{-i})$ are not known to the algorithm. Nevertheless, we can use \emph{approximations} to those values through sampling (similar to Theorem \ref{theo:NEdist}). In the proof of the theorem we show that the above construction can be done even if we use approximate values for $u_i(a_i,x_{-i})$ rather than exact ones.

%
%
%
%
%
%
%

\subsection{The approximate fixed-point problem}

We will consider the $||\cdot ||_\infty$ norm on $\mathbb{R}^n$.  Thus, a function $f:[0,1]^n \rightarrow [0,1]^n$ is $\lambda$\emph{-Lipschitz} if $||f(x)-f(y)||_\infty \leq \lambda ||x-y||_\infty$. A point $x\in [0,1]^n$ is an $\varepsilon$\emph{-fixed point} of $f$ if $||f(x)-x||_\infty \leq \varepsilon$. The approximate fixed-point problem is defined as follows.

$\fp(n,\lambda,\varepsilon)$:

INPUT- $\lambda$-Lipschitz function $f:[0,1]^n \rightarrow [0,1]^n$.  

OUTPUT- $\varepsilon$-fixed point of $f$.

QUERIES- Each query is a point $x\in [0,1]^n$ and the answer is $f(x)$.

The query complexity of the approximate fixed-point problem  was studied by Hirsch, Papdimitriou, and Vavasis \cite{HPV} in \emph{deterministic} settings. As was mentioned above, in the proof of Theorem \ref{theo:main} we reduce the problem of an approximate WSNE to the problem of an approximate fixed point. Then we prove that the approximate fixed-point problem requires an exponential number of queries, even if we allow \emph{probabilistic} algorithms. This result generalizes the result in \cite{HPV} to probabilistic algorithms, and may be of independent interest.

\begin{corollary}
Fix $\lambda=79$ and $\varepsilon=1/88$, then for $p=2^{-\Omega(n)}$ holds\newline
$QC_p(\fp(n, \lambda, \varepsilon))= 2^{\Omega(n)}$. 
\end{corollary}

Moreover, in probabilistic settings (unlike deterministic settings) we can use the sampling method to derive lower bounds for the case where the queries are \emph{distributions}. This observation answers the open question that was presented in \cite{HPV}: what is the query complexity of finding an approximate fixed point of a function $f:[0,1]^n \rightarrow [0,1]^n$ when every query is a distribution over $[0,1]^n$?

Note that exactly as in the Nash equilibrium case (see the beginning of Section \ref{sec:dist}), the case where the answer to a distribution $\mu$ is the exact value $\E_{x\sim \mu}f(x)$ is not interesting. By similar arguments to those that appear in Example \ref{ex}, we can use only one query to get the values of $f$ on an arbitrary small grid, if we know that the values of the function on this grid are rational numbers with a denominator at most $M$. This is indeed the case, for instance, if the function is \emph{rational} (a quotient of two polynomials). Therefore, as in Section \ref{sec:dist}, we analyze the problem when the answers are given with a precision $\delta$.

$\fpd(n,\lambda,\varepsilon,\delta)$:

INPUT- $\lambda$-Lipschitz function $f:[0,1]^n \rightarrow [0,1]^n$.  

OUTPUT- An $\varepsilon$-fixed point of $f$.

QUERIES- Each query is a distribution $\mu$ over $[0,1]^n$ and the answer is some vector $w$ where $||w-\E_{x\sim \mu} f(x)||_\infty<\delta$.

\begin{theorem}\label{theo:fp}
Fix $\lambda=79$ and $\varepsilon=1/88$, then for $p=2^{-\Omega(n)}$ holds\newline
$QC_p(\fpd(n,\lambda,\varepsilon,\delta))\geq \delta^2 2^{\Omega(n)}$.
\end{theorem}

Note that even if the answers are given with an exponentially small precision ($\delta=2^{-\Omega(n)}$), the exponential lower bound still holds.

The proof is exactly the same as the proof of Theorem \ref{theo:NEdist}: we can implement every $\fpd$ algorithm by an $\fp$ algorithm using the sampling method. 

\subsection{Adaptive dynamics}\label{sec:dyn}

One of the central tools to derive lower bounds on the rate of convergence of adaptive dynamics is the \emph{communication complexity} tool (see \cite{KN}). Conitser and Sandholm \cite{CS} first introduced this idea in their study of two-player games. Later, Hart and Mansour \cite{HM} studied the communication complexity of Nash equilibria in $n$-player games. Hart and Mansour \cite{HM} showed that the communication complexity of exact pure and mixed Nash equilibria is exponential in $n$. As a consequence, they derived that there exists no \emph{uncoupled} dynamic (see \cite{HMc2} and \cite{HMc3} for definition and discussion on uncoupled dynamics) that converges to a pure or an exact mixed Nash equilibrium faster than $\exp(cn)$ steps (for constant $c$). The question regarding the communication complexity of and the rate of convergence to an \emph{approximate} Nash equilibrium, however, remained an open question.

Here we will not address the question of the communication complexity of an approximate Nash equilibrium. The query complexity model is weaker than the communication complexity model. Nevertheless, our result on the query complexity does induce interesting insights into the rate of convergence to an approximate Nash equilibrium of adaptive dynamics.

The communication complexity model induces results on the important class of uncoupled dynamics. The query complexity model induces results on a different class of dynamics, which we will call $k$\emph{-queries dynamics}. As we will see, this class of dynamics contains most of the known adaptive dynamics.

\subsubsection{Dynamics model and $k$-queries dynamics}

We introduce very brief description of the dynamic model.

In the dynamic settings we assume that the same one shot game $u$ is played repeatedly over time $t=1,2,...$. A \emph{history of play at time $t$} is the sequence $h(t)=(a(1),a(2),...,a(t-1))$ of past realized action profiles. For general dynamic, the mixed action of every player $i$ at time $t$ depends on the game $u$ and on the history of play at time $t$, and will be denoted by $x_i(t)=s_i(u,h(t))$. The realized pure action profile $a(t)$ is drawn according to the mixed action profile $(x_i(t))_{i\in [n]}$. The \emph{dynamic} is specified by the mappings $(s_i)_{i\in [n]}$, where $s_i$ maps every payoff function and history of play to player's $i$ next mixed action.

The idea in the definition of $k$-queries dynamics is to ask: How many additional payoff queries are needed to calculate the mixed strategies $x_i(t)=s_i(u,h_t)$ of all player $i$ at time $t$? Where by ``additional" we mean ``additional to the queries that was already asked until time $t-1$".

Let as illustrate this idea by an example.

\begin{example}\label{ex:rel}
Consider the class of \emph{regret based dynamic}, i.e., dynamics where at each time $t$, the mixed action of every player $i$ is a function of the regrets $\{R^i_{a_i\rightarrow a'_i}: a_i,a'_i\in A_i \}$ , where the regrets are calculated according to the aggregate joint action of the opponents until time $t$. See \cite{HMc4} for the definition of regrets and a discussion on regret-based dynamics. Note that all the regrets of player $i$ at time $t$ depend only on the payoffs $\{u_i(a_i,a(t')_{-i}): a_i\in A_i, t'<t\}$. Therefore, in order to calculate the regrets of player $i$ it is sufficient at each time $t$ to query the $m$ actions $\{(a_i,a(t-1)_{-i}):a_i\in A_i\}$ (note that the actions $\{(a_i,a(t')_{-i}):a_i\in A_i, t'<t-1\}$ was already queried in the previous steps). Hence, we will say that regret based dynamics are $nm$\emph{-queries dynamics}, because $nm$ additional payoff queries at each step are sufficient to calculate the mixed strategy of all players. 
\end{example}

The formal definition is as follows:

\begin{definition}\label{def:qd}
A dynamic will be called $k$-queries dynamic, if there exists a mapping that assigns to each history of play a set of $k$ (additional) pure actions payoff queries, such that the mixed strategy of all players at time $t$ can be calculated using the $tk$ queries until time $t$.
\end{definition}

\subsubsection{The generality of $k$-queries dynamics}

We argue that most studied adaptive dynamics are $mn$-queries dynamics.

Example \ref{ex:rel} illustrates the fact that \emph{regret based} dynamics are $mn$-queries dynamics. The arguments in the example can be applied not only to the regret matching dynamic, but also to many the other studied regret minimizing dynamics as eigenvector dynamics, smooth fictitious play, and joint strategy fictitious-play (an overview of regret minimizing dynamics appears in \cite{HMc4}).

\emph{Better reply dynamics} are also $mn$-queries dynamics, because the mixed action of every player $i$ at time $t$ depends on the set of payoffs $\{u_i(a_i,a(t-1)):a_i\in A_i\}$, which again requires $mn$ queries. The class of better reply dynamics includes important dynamics as best-reply dynamic and logit dynamic (see \cite{B}).

Another class of dynamics that was studied in the literature is \emph{evolutionary dynamics}, as for example replicator dynamics, and smith dynamic (an overview of regret minimizing dynamics appears in \cite{S}). Evolutionary dynamics in population games are a specific case of better-reply dynamics, and therefore they are also $mn$-queries dynamics.

\subsubsection{Lower bound on the rate of convergence}

Clearly every $k$-queries dynamic that converges to an approximate equilibrium (or any other solution concept of the game) in $T$ steps with probability $p$ induces a query algorithm in the pure-query model that finds an approximate equilibrium in at most $kT$ queries with probability $p$.

Therefore, from Theorem \ref{theo:bin} we get the following corollary regarding the rate of convergence of $k$-queries dynamics:

\begin{corollary}\label{cor:dyn-pure}
There is no $k$-queries dynamic that converges to an $\varepsilon$-well-supported Nash equilibrium in $2^{\Omega(n)}/k$ steps with probability of at least $2^{-\Omega(n)}$ in all $n$-player binary action games.
\end{corollary}

By the Minmax Theorem (which is also called in this context Yaho's minmax Theorem) the impossibility result can also be extended to the Bayesian settings, where the game is drawn according to a probability distribution.

\begin{corollary}\label{cor:dyn-bays}
There exists a distribution over $n$-players binary-actions games, such that for every $k$-queries dynamic the expected\footnote{The expectation of the number of steps is taken over the game and the probabilistic process induced by the dynamic (in case the dynamic is not deterministic).} number of steps until the dynamic converges to a $\varepsilon$-well-supported Nash equilibrium is at least $2^{\Omega(n)}/k$.
\end{corollary}



Corollary \ref{cor:dyn-bays} states that there exists a distribution over games which is universally bed instances for all dynamics, whereas Corollary \ref{cor:dyn-pure} states that for every dynamic there exists some bad instances.

There are important dynamics, as for example the fictitious-play dynamic (see \cite{R}), where in order to calculate the mixed action of the player using small number of queries we must use \emph{mixed} action queries, rather than pure. In fictitious-play for example the strategy of every player $i$ at time $t$ is determined by the payoffs $\{u_i(a_i,(g_j)_{j\neq i}):a_i\in A_i\}$, where $g_j$ is the aggregate behavior of player $j$ up to time $t$. For those dynamics we can similarly define $k$\emph{-mixed-queries} dynamics, or more general class of $k$\emph{-distribution-queries} dynamics. In order to derive a lower bound on the rate of convergence of those dynamics we should rely on the impossibility result of Theorem \ref{theo:NEdist}, which corresponds to the case where the queries could be mixed actions, or more general-- distributions over action profiles. A difficulty arises when we try to do so. The query model in Theorem \ref{theo:NEdist} assumes that the answers are given with a precision $\delta$, whereas in the dynamics model the answers are precise. An additional assumption of \emph{continuity of the dynamic} (see Definition \ref{def:con}), makes it possible to derive a query algorithm with \emph{approximate} answers using a dynamic that depends on the \emph{exact} answers. The formal discussion on distribution-queries dynamics, and the lower bound on the rate of convergence of continuous $k$-distribution-queries dynamics (see Theorem \ref{theo:dyn}) are relegated to Appendix B.

\section{Proof of Theorem \ref{theo:main}}\label{sec:proof}

\subsection{From Approximate Nash Equilibrium to Approximate Fixed Point}\label{sec:ne->fp}

\subsubsection{Games with constant number of actions}\label{sec:ne->fp-const}

We show a reduction from the $\fp(n,\lambda, \varepsilon)$ problem to the $\ne(2n,k+1,\frac{3}{4k^2})$ problem, where $k=\lceil \frac{ \lambda +3}{2\varepsilon}\rceil$ (see Sections \ref{sec:main} and \ref{sec:fp} for the definitions of the problems). The reduction is based on a construction that proves Brower's fixed-point theorem using Nash's theorem, and appears in a blog of Eran Shmaya \cite{S}.

Given a $\lambda$-Lipschitz function $f:[0,1]^n \rightarrow [0,1]^n$, we construct a game with two groups of $n$ players. The action set of all players is $\{0,\frac{1}{k},\frac{2}{k},...,1\}$ for $k=\lceil\frac{ \lambda +3}{2\varepsilon} \rceil$. We denote by $a=(a_i)_{i=1}^n / b=(b_i)_{i=1}^n$ the vector that is played by the first/second group of players. The payoff function of player $i$ in the first group is defined by 
\begin{equation}\label{eq:payoff-def1}
u_i(a,b)=-|a_i-b_i|^2.
\end{equation}
The payoff function of player $i$ in the second group is defined by
\begin{equation}\label{eq:payoff-def2}
v_i(a,b)=-|b_i-f_i(a)|^2,
\end{equation}
where $f_i$ denotes the $i$th coordinate of the function $f$.

Simply speaking, the first group is trying to match the vector of the second group, whereas the second group is trying to match the $f$ operation on the vector of the first group.

Let $(x,y)$ be a $\frac{3}{4k^2}$-WSNE of the game, where $x=(x_i)_{i=1}^n/y=(y_i)_{i=1}^n$ is the mixed-actions profile of the players in the first/second group.

The payoff of player $i$ in the first group that faces the mixed strategy $y$ of the second group can be written as
\begin{equation}\label{eq:pl1payoff}
u_i(a_i,y)=-(a_i-\E (y_i))^2 -Var (y_i).
\end{equation}

Let $\alpha_i\in \mathbb{N}$ be such that $\frac{\alpha_i}{k} \leq \E (y_i) \leq \frac{\alpha_i+1}{k}$. W.l.o.g. we assume that $\E (y_i) \leq \frac{\alpha_i+0.5}{k}$; i.e., $\E (y_i)$ is closer to $\frac{\alpha_i}{k}$ than to $\frac{\alpha_i+1}{k}$. By equality (\ref{eq:pl1payoff}) it is clear that $\frac{\alpha_i}{k}$ is a best reply of player $i$ and his payoff at a best reply is at least
\begin{equation}\label{ineq:pay1}
u_i(\frac{\alpha_i}{k}, y)\geq -\frac{1}{4k^2} - Var(y_i).
\end{equation}
For every action $\frac{\gamma}{k}$ where $\gamma\neq \alpha_i,\alpha_i+1$, player $i$'s payoff is at most 
\begin{equation}\label{ineq:pay2}
u_i(\frac{\gamma}{k}, y)\leq -\frac{1}{k^2} - Var(y_i).
\end{equation}

Therefore in a $\frac{3}{4k^2}$-WSNE player $i$ assigns positive probability only to the strategies $\frac{\alpha_i}{k}$ and $\frac{\alpha_i+1}{k}$.

Following similar arguments for the second group of players we write player's $i$ payoff as
\begin{equation*}
v_i(x,b_i)=-(b_i-\E (f_i(x)))^2 -Var (f_i(x)),
\end{equation*}
and we derive that player $i$ assigns positive probability only to the strategies $\frac{\beta_i}{k}$ and $\frac{\beta_i+1}{k}$ where $\frac{\beta_i}{k} \leq \E (f_i(x)) \leq \frac{\beta_i+1}{k}$.

For every approximate WSNE we set $c_i=\frac{\alpha_i+0.5}{k}$, and $d_i=\frac{\beta_i+0.5}{k}$. We claim that the point $c=(c_i)_{i\in [n]}$ is an approximate fixed point of $f$. 


For every $i$ we have
\begin{equation*}
|c_i-d_i| \leq |c- \E (y_i)| + |\E (y_i) -d_i| \leq \frac{0.5}{k}+\frac{0.5}{k}
\end{equation*}
and
\begin{equation*}
|d_i-f_i(c)| \leq |d_i- \E (f_i(x))| + | \E (f_i(x)) - f_i(c)| \leq \frac{0.5}{k}+\lambda \frac{0.5}{k}.
\end{equation*}
Therefore $|c_i-f_i(c)| \leq \frac{\lambda+3}{2k} \leq \varepsilon$, which implies that $||c-f(c)||_\infty \leq \varepsilon$.

This construction yields the following result.

\begin{proposition}\label{pro:ne->fp}
For every $p>0, n\geq 1 , \lambda \geq 0$, and $\varepsilon \geq 0$, set $k=\lceil \frac{ \lambda +3}{2\varepsilon}\rceil$. Then we have
\begin{equation*}
QC_p(\ne(2n, k+1, \frac{3}{4k^2})) \geq QC_p(\fp(n,\lambda,\varepsilon)).
\end{equation*}
\end{proposition}

\begin{proof}
Let $\mathcal{F}$ be the set of all $\lambda$-Lipschitz functions $f:[0,1]^n \rightarrow [0,1]^n$. Let $\mathcal{U}$ be the set of games that correspond to $\mathcal{F}$ by the above-presented construction.

Every $\ne$ algorithm with success probability $p$ on the set of games $\mathcal{U}$ induces an $\fp$ algorithm with success probability $p$. The induced $\fp$ algorithm will follow the $\ne$ algorithm. Each query $(a,b)$ will be mapped to the query $f(a)$. Note that by equations (\ref{eq:payoff-def1}) and (\ref{eq:payoff-def2}), $f(a)$ is sufficient to get the answer for the payoff profile $(u_i,v_i)$. Therefore the $\fp$ algorithm can indeed follow the $\ne$ algorithm. Finally, given a $\frac{3}{4k^2}$-WSNE the algorithm can compute the values $(\alpha_i)_{i\in [n]}$ and to find the approximate fixed point $c=(c_i)$ of $f$.
\end{proof}

\subsubsection{Binary-action games}\label{sec:ne->fp-bin}

The idea is to use similar idea to those presented above (Section \ref{sec:ne->fp-const}). Here, instead of $2n$ players with $k+1$ actions for each player, we construct a game with $2n(k-1)$ players with two actions for each player.
In the above reduction the action set of each player $i$ is $\{0,\frac{1}{k},\frac{2}{k},...,1\}$. We replace every player $i$ (in each one of the two groups) by $k-1$ agents  $(i,\frac{1}{k}),(i,\frac{2}{k}),...,(i,\frac{k-1}{k})$. The action set of agent $(i,\frac{j}{k})$ is $A_{i,j}:=\{\frac{j-1}{k},\frac{j+1}{k}\}$. For simplicity of notations, it will be convenient to have two additional agents: agent $(i,\frac{0}{k})$ with single action $A_{i,0}=\{\frac{1}{k}\}$, and agent $(i,\frac{k}{k})$ with single action\footnote{Since these two additional agents have a single action they are not counted as players.} $A_{i,k}=\{\frac{k-1}{k}\}$. When agent $(i,\frac{j}{k})$ plays the action $(i,\frac{j\pm 1}{k})$ we will say that agent $(i,\frac{j}{k})$ is \emph{ pointing on player} $(i,\frac{j\pm 1}{k})$. 
 
For every $i$, given an action profile of the agents $(i,\frac{j}{k})_{j=0}^k$, the \emph{realized value of the $i$th agents} is defined to be $r_i:=\frac{c+0.5}{k}$ where $(c,c+1)$ is the minimal pair of agents that pointing on each other. In other words, the minimal $c$ such that $(i,\frac{c+1}{k})$ plays $\frac{c}{k}$). Note that the realized value is well defined because the last agent $\frac{k}{k}$ always play $\frac{k-1}{k}$.

The payoff functions of the agents are defined similar to the previous reduction (Section \ref{sec:ne->fp-const}), but with respect to the \emph{realized values}.

The payoff of agent $(i,\frac{j}{k})$ for a player $i$ in the first group is defined by $u_{i,\frac{j}{k}}=-|a_{i,\frac{j}{k}}-r_i|^2$, where $r_i$ is the realized value of the $i$-agents in the second group. Similarly, the payoff of agent $(i,\frac{j}{k})$ for a player $i$ in the second group is defined by $v_{i,\frac{j}{k}}=-|a_{i,\frac{j}{k}}-f(r)|^2$, where $r$ is the profile of the realized value of all agents in the first group.

Every mixed action profile of the agents in the second group $y$ induces a distribution $\rho$ on the realized values of the second group. The payoffs of agent $(i,\frac{j}{k})$ that facing a mixed action profile can be written as
\begin{equation*}
u_{i,\frac{j}{k}}(a,y)=-(a_{i,\frac{j}{k}}-\mathbb{E}(\rho_i))^2-Var(\rho_i).
\end{equation*}
Therefore, the difference in payoffs for the two possible actions of agent $(i,\frac{j}{k})$ is given by
\begin{equation}\label{eq:dif}
  \begin{aligned}
d_{i,j}&=u_{i,\frac{j}{k}}\left(\frac{j-1}{k},y\right)-u_{i,\frac{j}{k}}\left(\frac{j+1}{k},y\right)\\
&=-\left(\frac{j-1}{k}-\mathbb{E}(\rho_i)\right)^2+\left(\frac{j+1}{k}-\mathbb{E}(\rho_i)\right)^2 \\
&=\left(\frac{j+1}{k}-\frac{j-1}{k}\right)\left(\frac{j+1}{k}+\frac{j-1}{k}-2\mathbb{E}(\rho_i)\right)=\frac{4}{k}\left(\frac{j}{k}-\mathbb{E}(\rho_i)\right)
   \end{aligned}
\end{equation}

In every $\frac{1}{k^2}$-WSNE the mixed actions of the $i$th agents satisfy:

(1) Every agent $(i,\frac{j}{k})$ such that $\frac{j}{k} < \mathbb{E}(\rho_i)-\frac{1}{4k}$ plays the action $\frac{j+1}{k}$ with probability 1. This follows from equation \eqref{eq:dif}, because in such a case we have $d_{i,j}<-\frac{1}{k^2}$, which implies that the action $\frac{j+1}{k}$ is better than the action $\frac{j+1}{k}$ by at least $\frac{1}{k^2}$.

(2) Similarly, every agent $(i,\frac{j}{k})$ such that $\frac{j}{k} > \mathbb{E}(\rho_i)+\frac{1}{4k}$ plays the action $\frac{j+1}{k}$ with probability 1.

Let $c_i$ be the closest integer multiple of $\frac{1}{k}$ to $\mathbb{E}(\rho_i)$ (formally, $c_i=\frac{[k\mathbb{E}(\rho_i)]}{k}$, where $[x]$ is the closest integer to $x$). By observations (1) we obtain that every agent $(i,\frac{j}{k})$ for $\frac{j}{k}<c_i$ plays $\frac{j+1}{k}$. By observations (2) we obtain that every agent $(i,\frac{j}{k})$ for $\frac{j}{k}>c_i$ plays $\frac{j-1}{k}$. Therefore the realized value of the $i$th agents is either $r_i=c_i-0.5$ or $r_i=c_i+0.5$, and in any case $|r_i-\mathbb{E}(\rho_i)|\leq \frac{1}{k}$.

By repeating the same arguments for the second group of players we obtain that in every $\frac{1}{k^2}$-WSNE the realized value $s_i$ of the $i$th agents satisfies $|s_i-\mathbb{E}_{r\sim \omega}(f_i(r))|\leq \frac{1}{k}$ where $\omega$ is the distribution of the realized value profiles for the first group.

Let $(\rho,\omega)$ be the distributions over the realized values profiles of both groups in an $\frac{1}{k^2}$-WSNE. Let $(r,s)$ be a profile of realized values in the support of $(\rho,\omega)$. We claim that $r$ is an approximate approximate fixed point of $f$.

\begin{align*}
|r_i-s_i|\leq |r_i-\mathbb{E}(\rho_i)|+|\mathbb{E}(\rho_i) - s_i| \leq \frac{1}{k}+ \frac{0.5}{k} 
\end{align*}
and
\begin{align*}
|s_i-f_i(r)|\leq |s_i-\mathbb{E}_{r\sim \omega}(f_i(r))|+|\mathbb{E}_{r\sim \omega}(f_i(r)) - s_i| \leq \frac{1}{k}+ \lambda \frac{0.5}{k} 
\end{align*}
Therefore, $|r_i-f_i(r)|\leq \frac{5+\lambda}{2k}$.

If we set $k=\lceil\frac{5+\lambda}{2\varepsilon} \rceil$ we have that $|r_i-f_i(r)|\leq \varepsilon$.

This construction yields the following result.

\begin{proposition}\label{pro:ne->fp-bin}
For every $p>0, n\geq 1, \lambda \geq 0$, and $\varepsilon \geq 0$, set $k=\lceil \frac{ 5+ \lambda }{2\varepsilon}\rceil$. Then we have
\begin{equation*}
QC_p(\ne(2n(k-1), 2, \frac{1}{k^2})) \geq QC_p(\fp(n,\lambda,\varepsilon)).
\end{equation*}
\end{proposition}

The proof is similar to the proof of Proposition \ref{pro:ne->fp}.

\subsection{From Approximate Fixed Point to End of Path}\label{sec:fp->ep}

We denote by $G(n,k)$ the graph of the $n$-\emph{dimensional grid of size} $k$. The set of vertices of $G(n,k)$ is $(c_i)_{i=1}^n$, where $c_i\in [k]$. There is an edge between $(c_i)_{i=1}^n$ and $(c'_i)_{i=1}^n$ iff $c_i=c'_i$ for all indexes $i$ except for one index $j$ for which $|c_j-c'_j|=1$. Note that $G(n,2)$ is the hypercube. A path on a graph will be called \emph{simple} if it contains no cycles.

In this section we prove a reduction from the approximate fixed point problem $\fp(n,79,1/88)$ to the end-of-a-simple-path problem:

$\ep(n)$:

INPUT- Simple path on the $n$-dimensional hypercube ($G(n,2)$) that starts at $(1,1,...,1)$.

OUTPUT- The end-of-path vertex.

QUERIES- Each query is a vertex $v$ of the hypercube. The answer is whether the path visits this vertex. If it does, then, in addition, the black box reports the path's previous and next visits.

The query complexity of an approximate fixed point has been studied by Hirsch, Papadimitriou and Vavasis \cite{HPV}, where they show an exponential (in the dimension) lower bound. We cannot use the result of \cite{HPV} straightforwardly, because they consider \emph{deterministic} settings, whereas we are interested in \emph{probabilistic} settings. Nevertheless, we will rely on one part of their proof. Namely, we will use their reduction from the approximate fixed-point problem to the end-of-a-simple-path problem. The formal treatment of this reduction is quite involved; a 13-pages proof of the reduction appears in \cite{HPV}. We present here only a brief intuition for the reduction and the result itself.

\begin{figure}[h]
\begin{center}
\caption{A path in $G(2,3)$.}\label{pic:path}
\includegraphics[scale=0.6]{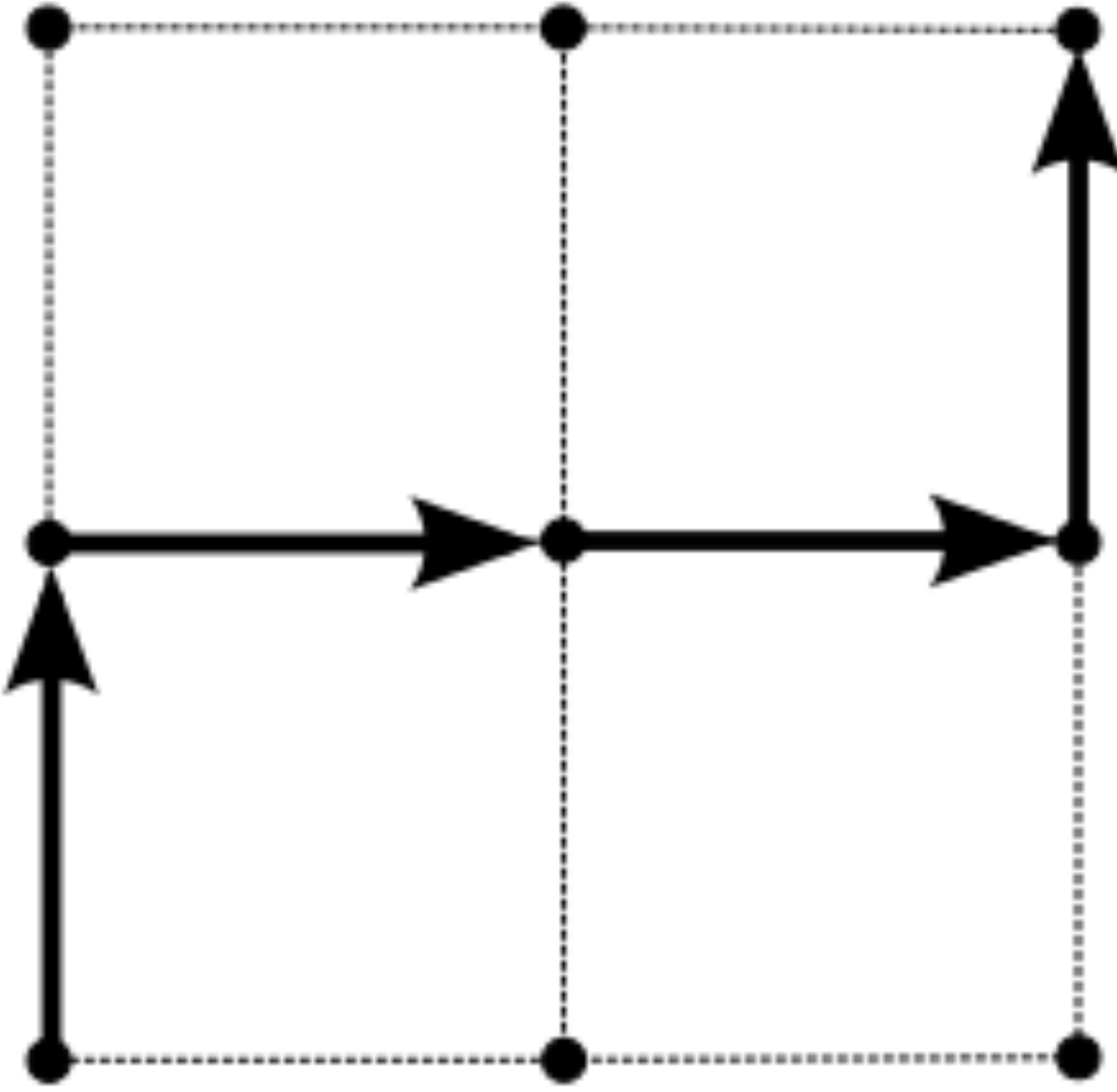}
\end{center}
\end{figure}

\begin{figure}[p]
\begin{center}
\caption{the function that corresponds to the path in Figure \ref{pic:path}. Each arrow from point $a$ to point $b$ represents $f(a)=b$. The top-right $3\times 3$ squares corresponds to the nodes in $G(2,3)$. The rest of the squares are required to complete the definition of the function. The black dot denotes the unique \emph{exact} fixed point of the function. Approximate fixed points are represented by short arrows.}\label{pic:arrows}
\includegraphics[scale=0.25]{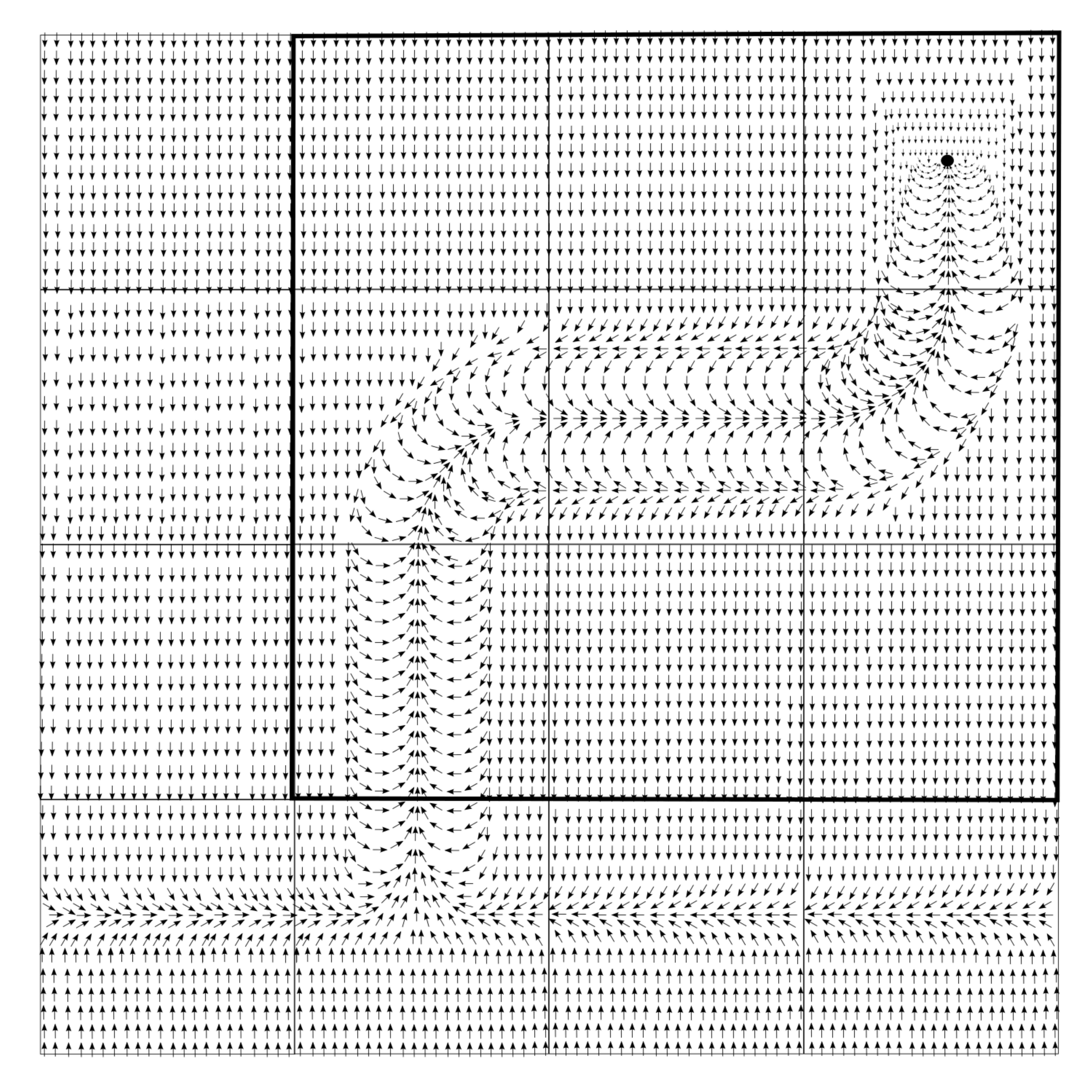}
\end{center}
\end{figure}

We divide the $n$-dimensional cube $[0,1]^n$ into small cubes of edge size $\delta=\Theta(1/k)$. The small cubes with the neighboring relation form the graph $G(n,k)$. Let us consider a specific example that explains the idea of the construction. Given a path on $G(2,3)$ with no cycles and with the starting point $(1,1)$, for example, the path presented in Figure \ref{pic:path}, we define a function $f$ on small cubes of size $\delta=1/4$ as demonstrated in Figure \ref{pic:arrows}. By observing Figure \ref{pic:arrows}, the reader may be convinced that the function $f$ has the following properties:

\begin{enumerate}
\item $f$ is Lipschitz continuous. This corresponds to the fact that the arrows are changed in a smooth manner.

\item All the approximate fixed points of $f$ are placed on the end-of-path square. This corresponds to the fact that all the arrows outside the top-right square have constant length, i.e., $f(x)$ is far from $x$ for every $x$ that is not in the top-right square.

\item For every small square, $f$ depends only on the following parameters of the path, namely, whether the path goes through this square, and if so, what the previous and the next visits of the path are. This corresponds to the locality property of the  picture. If the path does not go through a square, then all the arrows in this square point downward. If the path goes though some square, it is enough to consider the previous and the next visits of the path in order to define $f$ for this square.
\end{enumerate}

In \cite{HPV}, Hirsch et al. show that the construction of such a function $f$ can be generalized to every grid size $k$ and to every dimension $n$. In particular, they introduce a general construction that works for every approximation accuracy $\varepsilon$ and every Lipschitz constant $\lambda$. For our purposes, it will be sufficient to focus on the case\footnote{The cube size $1/6$, approximation accuracy $\varepsilon=1/88$, and Lipschitz constant $\lambda=79$ were chosen so that the necessary inequalities for the construction in \cite{HPV} would be satisfied. Explicitly, the necessary inequalities are $\delta \geq \frac{\lambda}{1200 \varepsilon}$ (see \cite{HPV} page 406), and $2\geq \frac{(1-10 \varepsilon)\lambda}{1200\varepsilon}-3$ (see \cite{HPV} page 411).}  $\varepsilon=1/88$, $\lambda=79$. For these values the construction in \cite{HPV} states the following:

Given a simple path on $G(n,2)$ that starts at the point $(1,1,...,1)$, and ends at the point $(e_i)_{i=1}^n$ for $e_i=1,2$, we divide the cube $[0,1]^n$ into small cubes with edge cube size\footnotemark[\value{footnote}] $\delta=1/6$. The vertex $(v_i)_{i=1}^n$ in the hypercube corresponds to the small cube $\times_i [\frac{1+v_i}{6}, \frac{2+v_i}{6}]\subset [0,1]^n$. By this correspondence, the hypercube is equivalent to all the small cubes that are contained in $[2/6,4/6]^n$. The starting point $(1,1,...,1)$ is equivalent to the small cube $[2/6,3/6]^n$, and the end point corresponds to the small cube $\times_i [\frac{1+e_i}{6}, \frac{2+e_i}{6}]$, which is denoted by $E$. 

\begin{remark}\label{rem:pro}
Hirsch et al. \cite{HPV} construct a function $f:[0,1]^n \rightarrow [0,1]^n$ with the following properties:

\begin{description}
\item[(P1)] $f$ is 79-Lipschitz (Lemma 10 in \cite{HPV}).

\item[(P2)] $||f(x)-x||_\infty \geq 1/88$ for every $x \notin E$ (Lemma 10 in \cite{HPV}).

\item[(P3)] The value of $f$ on each small block in $[2/6,4/6]^n$ depends only on the local behavior or the path in the corresponding vertex of the hypercube: whether the path goes through this vertex, and if so what the previous and the next visits of the path are (Lemma 11 in \cite{HPV}).
\end{description}
\end{remark}

This construction yields the following result:

\begin{proposition}\label{pro:fp->ep}
$QC_p(\fp(n,79,1/88)) \geq QC_p(\ep(n))$ for every $n\geq 1$ and every $0<p<1$. 
\end{proposition}

\begin{proof}
Let $\mathcal{P}$ be the set of all simple paths on the hypercube with starting point $(1,1,...,1)$. Let $\mathcal{F}$ be the corresponding set of functions $f:[0,1]\rightarrow [0,1]$ according to the construction of \cite{HPV} presented above.

Every $\fp$ algorithm with success probability $p$ on the set of functions $\mathcal{F}$ induces an $\ep$ algorithm with success probability $p$. The induced $\ep$ algorithm will follow the $\fp$ algorithm. Each query $x\in [2/6,4/6]^n$ will be mapped to a query of the corresponding vertex. For formal purposes, we should define also the map for queries $x\notin [2/6,4/6]^n$. Clearly those are meaningful queries because all the functions $f\in \mathcal{F}$ have the same values for $x\notin [2/6,4/6]^n$. For the sake of formality, we say that a query $x\notin [2/6,4/6]^n$ corresponds to the query $(1,1,...,1)$. By property (P3) the answer for the $\fp$ query contains at least as much information as the answer in the $\ep$ query. Therefore, in the $\ep$ model it is indeed possible to follow the $\fp$ algorithm. Finally, by property (P2), once the $\fp$ algorithm finds a $1/88$-fixed point, then the $\ep$ algorithm finds the end of path.
\end{proof}

\subsection{From End of Path to Hit the Path}\label{sec:ep->hp}

We would like to analyze the end-of-simple-path problem in a probabilistic setting. A similar analysis was recently conducted by Hart and Nisan \cite{HN}, who showed that the end-of-(general)-path problem is hard even in probabilistic settings. The difference between our problem and the problem in \cite{HN} is that we assume that the path has no cycles, whereas in \cite{HN} there is no such assumption. As we will see in this section, we can use similar arguments to derive a hardness result for our end-of-simple-path problem.

The idea in \cite{HN} is that finding the end of a random walk on the hypercube of length $2^{n/3}$ is hard. We cannot use directly the idea of a random path because the random walk has cycles with probability close to 1. Nevertheless, we can cut those cycles and get a simple path; this new path is no longer a random walk on the hypercube, but it does preserve the essential properties for the hardness result, as we will see below. Formally, given a random walk $v_1,v_2,...,v_{2^{n/3}}$, the \emph{path after cutting the cycles} is defined by the following iterative procedure: pick the minimal $i$ such that there exists $j>i$ where $v_i=v_j$. Replace the existing path by $v_1,v_2,...,v_i,v_{j+1},...,v_{2^{n/3}}$. Repeat this cutting process until no vertex appears twice in the path. We denote the resulting path without cycles by $w_1,w_2,...,w_L$, where $L$ is a random variable.

\begin{lemma}\label{lem:long-cyc}
Let $(v_i)_{i=1}^{2^{n/3}}$ be a random walk on the $n$-dimensional hypercube. Then $(v_i)$ contains no cycles of size greater than $n^2$ with probability of at least $1-2\cdot 2^{-n/3}$.
\end{lemma}

\begin{proof}
The mixing time of the random walk on the hypercube is known to be $O(n \log n)$ (see, e.g., \cite{DGM}). This implies that for every pair of times $i,j$ where $i<j+n^2$, the probability of $v_i=v_j$ is at most $2\cdot 2^{-n}$. Therefore, summing over all such pairs of $i,j$ (we have at most $2^{2n/3}$ pairs), we get that the probability that at least one of these events will happen is at most $2\cdot 2^{-n/3}$.  
\end{proof}

Now we proceed similar to \cite{HN}. In order to prove that there is no algorithm that finds $w_L$ with high probability, we define the following hit-the-path ($\hp(n)$) game between the algorithm and the adversary. The game is played for $T=2^{n/3}/n^4$ steps. The adversary chooses a path without cycles $w_1,..., w_L$. At each step $1\leq t \leq T$, the algorithm chooses a vertex $q_t$ when it observes the vertices revealed up to time $t-1$: $w_1,w_2,...,w_{n^2(t-1)}$, i.e., $q_t=q_t((w_i)_{i=1}^{n^2(t-1)})$. The goal of the algorithm is to choose a \emph{future vertex}, a vertex that is visited by the path $n^2$ steps later than the last revealed vertex ($w_{n^2(t-1)}$), i.e., $q_t=w_i$ for $i>tn^2$. After the algorithm chooses a vertex $q_t$, the vertices $w_{n^2(t-1)+1}, w_{n^2(t-1)+2},...,w_{n^2t}$ are revealed to the algorithm. The algorithm wins if it has succeeded to in choosing at least one future vertex. Otherwise the adversary wins.

\begin{lemma}
Using the mixed strategy $w_1,...,w_L$, which results from cutting the cycles of a random walk of size $2^{n/3}$, the adversary guarantees a win with probability of at least $1-3\cdot 2^{-n/3}$ in the $\hp(n)$ game.
\end{lemma}
  
\begin{proof}
Let $v_1,v_2,...,v_{2^{n/3}}$ be the random walk that induces the path $w_1,...,w_L$. For the bound analysis, we assume that if $(v_i)$ has at least one cycle of size greater than $n^2$, then the algorithm automatically wins. By Lemma \ref{lem:long-cyc} this yields a probability of at most $2\cdot 2^{-n/3}$ for the algorithm winning.

In the remaining case, let us change the roles of the game $\hp(n)$ in favor of the algorithm. First, we will provide the algorithm with more information at each step. Instead of revealing at each step $n^2$ sequential values of $(w_i)$, we will reveal $n^4$ sequential values of $(v_i)$. Note that $n^4$ sequential values of $(v_i)$ contain at least $n^2$ values of $(w_i)$ because $(v_i)$ has no cycles of size greater than $n^2$. Second, we will let the algorithm win not only if it hits a future vertex of $(w_i)$ but also if it hits a future vertex of $(v_i)$. 

Using again the mixing time of the random walk argument (see \cite{DGM}), we know that at each step the probability of hitting every future vertex $v_i$ is at most $2\cdot 2^{-n}$. Summing over all future vertices and over all steps we get that in the new game the probability that the algorithm will win is at most 
\begin{equation*}
2\cdot 2^{-n} 2^{\frac{n}{3}} \frac{2^{\frac{n}{3}}}{n^4} \leq 2^{\frac{n}{3}}. 
\end{equation*}

Summing this probability with the probability of automatic winning yields the result.

\end{proof}

\begin{proposition}\label{pro:ep-hard}
$QC_p(\ep(n))\geq 2^{n/3}/n^4$ for $p=3\cdot 2^{-n/3}$ and for every $n\geq 1$.
\end{proposition}

\begin{proof}
In order to show that every \emph{probabilistic} algorithm fails to find the end of the path in $T=2^{n/3}/n^4$ queries with probability of at least $p=1-3\cdot 2^{-n/3}$, by the minmax theorem it is enough to show that there exists a distribution over paths such that every \emph{deterministic} algorithm fails to find the end of the path with probability of at least $p$.

Our random path will be $w_1,...,w_L$, as was described above. It is enough to show that every algorithm $A_E$ for $\ep(n)$ with success probability $p$ induces a strategy $A_H$ in the game $\hp(n)$ with success probability $p$. The algorithm $A_H$ will just follow the algorithm $A_E$. If up to time $t$ the algorithm $A_H$ has not hit a future vertex of the path, then at time $t$ the $A_H$ algorithm can calculate all the answers to the $\ep(n)$ queries (because in the $\hp$ settings it has at least as much information as in the $\hp$ settings), and using those answers it can produce the next query $q_{t+1}$. Therefore it is indeed possible to follow the $A_E$ algorithm. 

Finally, since the end of the path is a future vertex for all steps $t$, finding the end of the path will guarantee winning in the $\hp(n)$ game.

\end{proof}

\subsection{Proof of the main Theorems}

The proof of Theorem \ref{theo:main} is obtained by joining the three reductions in Sections \ref{sec:ne->fp-const}, \ref{sec:fp->ep}, and \ref{sec:ep->hp}.

\begin{proof}[Proof of Theorem \ref{theo:main}]
By Propositions \ref{pro:ep-hard}, \ref{pro:fp->ep}, and \ref{pro:ne->fp} we get the result.

For $p=3\cdot 2^{-n/3}$ we have
\begin{align*}
2^{n/3}/n^4 & \leq QC_p(\ep(n)) \\
& \leq QC_p(\fp(n, 79, 1/88)) \\ 
& \leq QC_p(\ne(n, 3608+1, \frac{3}{4}3608^{-2})) \\
& \leq QC_p(\ne(n, 3609, 2^{-24})).
\end{align*}
\end{proof}

The proof of Theorem \ref{theo:bin} is obtained by joining the three reductions in Sections \ref{sec:ne->fp-bin}, \ref{sec:fp->ep}, and \ref{sec:ep->hp}.

\begin{proof}[Proof of Theorem \ref{theo:bin}]
By Propositions \ref{pro:ep-hard}, \ref{pro:fp->ep}, and \ref{pro:ne->fp-bin} we get that for every $n$ and $p=3\cdot 2^{-n/3}$ holds:
\begin{align*}
2^{n/3}/n^4 & \leq QC_p(\ep(n)) \\
& \leq QC_p(\fp(n, 79, 1/88)) \\ 
& \leq QC_p(\ne(7390n, 2, 3695^{-2})).
\end{align*}

By replacing $n$ with $\frac{n}{7390}$ we have that for $p=3\cdot 2^{-n/22170}$ holds 
\begin{align*}
10^{16} \cdot \frac{2^{n/22170}}{n^4} \geq QC_p(\ne(n, 2, 3695^{-2})).
\end{align*}
\end{proof}

%

\section{Open problems}

This paper presents one basic result on the complexity of an approximate Nash equilibrium in games with a large number of players $n$, where every player has a constant number of actions $m$ (or even just two actions). Even for these games, many questions still remain open.
\begin{enumerate}

\item For an approximate (not well supported) Nash equilibrium, this result yields an exponential lower bound of the query complexity, only for the case where the approximation is $\varepsilon=O(1/n)$ (Theorem \ref{theo:ane}). The case of an $\varepsilon$-Nash equilibrium for \emph{constant} $\varepsilon$ remains open. \emph{What is the query complexity of an $\varepsilon$-Nash equilibrium for constant $\varepsilon$?}

\item The result provides an exponential lower bound for constant but tiny approximation value, ($\varepsilon=\frac{1}{2}10^{-7}$). It will be interesting to improve the exponential lower bound for bigger values of $\varepsilon$. \emph{What is the query complexity of the problems $\ne(n,m,\varepsilon)$ for\footnote{For $\varepsilon=1/2$ it is known that the query complexity is polynomial; see \cite{DMP} and \cite{FGGS}.} $0<<\varepsilon<1/2$ and for constant $m$?}

\item As mentioned in the Introduction, from the computational complexity perspective this result provides evidence that for these games there is no algorithm for computing an approximate Nash equilibrium with running time $poly(n)$, or equivalently $poly(\log(N))$ where $N=nm^n$ is the input size. To the best of our knowledge, even $poly(N)$ algorithm is not known for this class of games. Lipton,  Markakis, and Mehta's \cite{LMM} sampling method provides an algorithm for computing an $\varepsilon$-Nash equilibrium with running time $poly(N^{\log N})$. Daskalakis and Papadimitriou \cite{DP} proved existence of $poly(N^{\log \log N})$ algorithm. Babichenko Barman and Peretz \cite{BBP} proves existence of $poly(N^{\log \log \log N})$ algorithm. \emph{Is there an algorithm that computes an $\varepsilon$-Nash equilibrium in $poly(N)$ steps?}

\item As mentioned in Section \ref{sec:dyn}, query complexity protocol is a special case of communication complexity protocol, i.e., the communication complexity lower bound induces a query complexity lower bound, but not vice versa. We want to emphasize that the \emph{communication complexity of an approximate (well-supported or not) Nash equilibrium} remains an open question.
\end{enumerate}

\section{Appendix A- Proof of Theorem \ref{theo:ane}}

\begin{proof}[Proof of Theorem \ref{theo:ane}]
We start with the following modification of Lemma 4.28 in \cite{DGP} on the construction of approximate WSNE from approximate NE.

Let $x=x_i$ be a $(\varepsilon^2/(16n))$-Nash equilibrium. For every player $i$ we classify the actions of player $i$ into three groups, according to the outcome of this action against $x_{-i}$ ($G_i$ \emph{good actions}, $M_i$ \emph{medium actions}, and $B_i$ \emph{bad actions}).
\begin{eqnarray*}
& G_i &=\{a_i\in A_i: br_i - \frac{\varepsilon}{4} \leq u_i(a_i,x_{-i})  \} \\
& M_i &=\{a_i\in A_i: br_i - \frac{\varepsilon}{2} < u_i(a_i,x_{-i})\leq br_i - \frac{\varepsilon}{4} \} \\
& B_i &=\{a_i\in A_i: u_i(a_i,x_{-i})< br_i - \frac{\varepsilon}{2} \}
\end{eqnarray*}
We fix some $g_i^*\in G_i$, and we let $y_i$ be any mixed action that moves \emph{all} the probability mass from $B_i$ to $g_i^*$, and in addition moves \emph{some} probability mass from $M_i$ to $g_i^*$ (in any possible way). Formally, $y_i(B_i)=0$, $y_i(m_i)\leq x_i(m_i)$ for every $m_i \in M_i$, $y_i(g_i)=x_i(g_i)$ for every $g_i\in G_i$ $g_i \neq g_i^*$, and finally $y_i(g_i^*)$ is defined so that the total measure of $y_i$ is 1.

Now we claim that every such profile $(y_i)_{i=1}^n$ is an $\varepsilon$-WSNE.

First, let us note that $x_i(M_i \cup B_i) \leq \varepsilon/(4n)$ because otherwise, only the losses from playing the actions in $M_i \cup B_i$ will be higher than $(\varepsilon^2/(16n))$. Therefore, the distance in total variation from $x_i$ to $y_i$ is at most $\varepsilon/(4n)$.

Second, let us note that if every opponent of player $i$ changes his mixed strategy with probability of at most $\varepsilon/(4n)$, then the payoff of player $i$ may change by at most $(n-1)\varepsilon/(4n)<\varepsilon/4$ (we recall that the payoffs are bounded in $[0,1]$).

Finally, every action that is played by $y_i$ with positive probability is an $\varepsilon/2$-best reply to $x_{-i}$. Therefore it is an $\varepsilon$-best reply to $y_{-i}$, because the best reply may increase by at most $\varepsilon/4$, whereas the performance of each action may decrease by at most $\varepsilon/4$.

Now, if we fix the threshold for classifying good and bad actions to be $br_i-3\varepsilon/8$, and we receive answers about $u_i(a_i,x_{-i})$ with precision $\varepsilon/8$ then indeed all the good actions will be classified to $G_i$, and all the bad actions will be classified to $B_i$. It will make no difference to us where the actions in $M_i$ are classified.

Now we are able to prove the reduction from an $\ane$ algorithm to a $\ne$ algorithm. We show that every $\ane(n,2,\varepsilon^2/(16n))$ algorithm that uses $T-64n^2/\varepsilon n^2$ queries and has success probability $p+2^{-n}$ induces an $\ne(n,2,\varepsilon)$ algorithm that uses $T$ samples and has success probability $p$. The $\ne$ algorithm follows the $\ane$ algorithm to find a $(\varepsilon^2/(16n))$-Nash equilibrium. Then it evaluates the numbers $u_i(a_i,x_{-i})$ using $32n/\varepsilon^2$ samples (the total number of samples will be $2n\frac{32n}{\varepsilon^2}$). By the Hoeffding inequality the probability that all the samples will approximate the values $u_i(a_i,x_{-i})$ with precision of $\varepsilon/8$ is at least $1-4ne^{-n}>1-2^{-n}$. Then the algorithm will use the above procedure with the threshold $br_i-3\cdot 10^{-8}/8$. If the $\ane$ algorithm indeed finds a $(\varepsilon^2/(16n))$-Nash equilibrium, and the samples indeed $\varepsilon/8$-approximate all the values $u_i(a_i,x_{-i})$ that accrue with probability of at least $p$, then the induced $\ne$ algorithm has succeeded in finding $\varepsilon$-WSNE. 

For a constant $\varepsilon$ we get from Theorem \ref{theo:bin} that every $\ane(n,2,\varepsilon^2/(16n))$ algorithm that uses $2^{\Omega(n)}-64n^2/\varepsilon^2=2^{\Omega(n)}$ queries has success probability of at most $2^{-\Omega(n)}+2^{-n}=2^{-\Omega(n)}$. In order to complete the proof, we observe that
\begin{align*}
2^{\Omega(n)} = 2^{\Omega(\frac{\varepsilon^2}{16} n)} \leq QC_p\left( \ane \left( \frac{\varepsilon^2}{16} n,2,\frac{1}{n} \right) \right) \leq QC_p(\ane( n,2,\frac{1}{n})).
\end{align*}
\end{proof}

\section{Appendix B- distribution-queries dynamics}

The class of $k$-distribution-queries dynamics is defined similarly to $k$-queries dynamics.

\begin{definition}
A dynamic will be called $k$\emph{-distribution-queries} dynamic, if there exists a mapping that assigns to each history of play a set of $k$ (additional) distribution payoff queries, such that the mixed strategy of all players at time $t$ can be calculated using the $tk$ queries until time $t$.
\end{definition}

By the definition of $k$\emph{-distribution-queries} dynamics, for every such dynamic the mixed strategy of player $i$ at time $t$ is $x_i(t)=f_i(u(d_1),...,u(d_{tk}))$ where $d_1,...,d_{tk}\in \Delta(A)$ are the distribution-queries that was asked until time $t$.

We use the total variation distance on $\Delta(A_i)$, i.e., for $x_i,y_i\in \Delta(A_i)$
\begin{align*}
d_1(x_i,y_i)=\sum_{j=1}^m |x_i(a_j)-y_i(a_j)|.
\end{align*}

\begin{definition}\label{def:con}
We will say that a $k$-distribution-queries dynamic is\newline $\nu$\emph{-Lipschitz continuous} if it is $\nu$-Lipschitz continuous with respect to the answers of the queries, i.e., 
$d_1(f_i(v_1,...,v_{tk}),f_i(w_1,...,w_{tk}))\leq \nu \alpha$ for every player $i$, every $t\in \mathbb{N}$ and every pair of $n$-dimensional vector sequences $(v_1,...,v_{tk}),(w_1,...,w_{tk})$ that satisfy $||v_l-w_l||_\infty \leq \alpha$ for every $l\in [tk]$.
\end{definition}

$\nu$ Lipschitz continuity of the dynamic, simply means that the mixed strategies of the players will not change by a lot if the answers to the payoff queries not changes by a lot.

Now we are able to state the version of Corollary \ref{cor:dyn-pure} for the case of $k$-distribution-queries dynamics.

\begin{theorem}\label{theo:dyn}
There exists constant $\varepsilon$, $k=2^{\Omega(n)}$ and $T=2^{\Omega(n)}$ such that there exists no $\nu$-Lipschitz continuous $k$-distribution-queries dynamic that converges to an $\varepsilon$-well-supported Nash equilibrium in $T$ steps with probability more than $\nu 2^{-\Omega(n)}$ in all $n$-player binary action games.
\end{theorem}

Note that the exponential lower bound holds even if the Lipschitz constant $\nu$ is exponentially large.

\begin{proof}[Proof of Theorem \ref{theo:dyn}]

Fix constant $\varepsilon$ and fix $T=2^{\Omega(n)}$, $k=2^{\Omega(n)}$, $p=2^{-\Omega(n)}$ $\delta=2^{-\Omega(n)}$ such that 
\begin{enumerate}
\item Every distribution query algorithm that uses $Tk$ distribution queries, and receives answers with precision $\delta$, finds an $\varepsilon$-WSNE in $n$-player binary-action games with probability of at most $p$. By Theorem \ref{theo:NEdist} such values $\varepsilon,T,k,p,$ and $\delta$ exist.

\item $\delta T=2^{-\Omega(n)}$. We can always guarantee this condition by reducing the constant at the exponent of $T$. Note that if condition (1) is satisfied then it is satisfied also after reducing the constant at the exponent of $T$.
\end{enumerate}

Given a $\nu$-Lipschitz dynamic $D$ that finds $\varepsilon$-WSNE in $T$ steps with probability $p+\nu\delta n T=p+\nu 2^{-\Omega(n)}$, we define a new class of dynamics $\mathcal{D}^\delta$ that receives the answers to the queries with noise $\delta$ (every different noise on the queries defines a different dynamic). The dynamic $D$ defines a random variable $h(T)$ over histories of play of size $T$. Similarly, every dynamic $D_\delta \in \mathcal{D}^\delta$ defines the random variable $h_\delta(T)$. By $\nu$-Lipschitz continuity, for every history $h$ the total variation distance between the mixed actions of each player (in $D$ and $D_\delta$) is at most $\nu\delta$. Therefore, the total variation of the mixed action profile is $\nu\delta n$ (we recall that total variation is a metric and we can use triangle inequality). Therefore, the total variation distance between the random variables $h(T)$ and $h_\delta(T)$ is at most $\nu\delta n T$. Therefore, if using $h(T)$ we can find an $\varepsilon$-WSNE with probability of at least $p+\nu\delta n T$, then using $h_\delta(T)$ we can find an $\varepsilon$-WSNE with probability of at least $p$.

Every class of $T$-steps $k$-distribution-queries dynamics $\mathcal{D}_\delta$ induces an algorithm that uses $Tk$ distribution queries with answer precision $\delta$. By condition (1) above, every algorithm solves the $\ned(n,2,\varepsilon,\delta)$ problem with probability of at most $p=2^{-\Omega(n)}$ after $Tk$ queries. Therefore, there is no dynamic $D$ that converges to $\varepsilon$-WSNE in $T$ steps with probability more than $p+\nu 2^{-\Omega(n)}=\nu 2^{-\Omega(n)}$.
\end{proof}

\end{document}